\documentclass[conference]{IEEEtran}
\usepackage{courier}
\usepackage[usenames,dvipsnames]{color} 
\usepackage{amssymb,amsmath}
\usepackage{graphicx}
\usepackage{colordvi}
\usepackage{epsfig}
\usepackage{endnotes}
\usepackage{verbatim}
\usepackage{subfigure}
\usepackage{textcomp}
\usepackage{subfig}
\usepackage{setspace}
\usepackage{xcolor}
\usepackage{amsmath}
\usepackage{cite}
\usepackage{epstopdf}
\usepackage{times}
\usepackage{tabu}
\usepackage{booktabs, multirow}

\usepackage{booktabs}
\usepackage[super]{nth}
\usepackage{amsthm}
\usepackage{algorithmicx,algpseudocode}
\usepackage{algpseudocode,algorithm,algorithmicx}
\usepackage{amssymb}
\usepackage{mathtools}

\usepackage{color}
\usepackage{soul}

\algrenewcommand\algorithmicrequire{\textbf{Precondition:}}
\algrenewcommand\algorithmicensure{\textbf{Postcondition:}}

\theoremstyle{definition}

\newtheorem{lemma}{Lemma}

\def\BibTeX{{\rm B\kern-.05em{\sc i\kern-.025em b}\kern-.08em
		T\kern-.1667em\lower.7ex\hbox{E}\kern-.125emX}}

\begin{document}
	\title{Cache Allocations for Consecutive Requests of Categorized Contents: Service Provider's Perspective} 
	
	\author{\IEEEauthorblockN{Minseok Choi$^{\dag\ast}$, Andreas F. Molisch$^{\dag}$, Dong-Jun Han$^{\ddagger}$, Joongheon Kim$^{\ast}$, and Jaekyun Moon$^{\ddagger}$}
		\IEEEauthorblockA{
			$^{\dag}$Ming Hsieh  Department of Electrical and Computer Engineering, University of Southern California, Los Angeles, USA \\
			$^{\ddagger}$School of Electrical Engineering, Korea Advanced Institute of Science and Technology (KAIST), Daejeon, South Korea \\
			$^{\ast}$School of Electrical Engineering, Korea University, Seoul, South Korea\\
			E-mails: 
			\texttt{choimins@usc.edu}, 
			\texttt{molisch@usc.edu}
			\texttt{djhan93@kaist.ac.kr},
			\texttt{joongheon@korea.ac.kr},\\
			\texttt{jmoon@kaist.edu}
		}
	}
	
	\maketitle
	
	\begin{abstract}
		In wireless caching networks, a user generally has a concrete purpose of consuming contents in a certain preferred category, and requests multiple contents in sequence. 
		While most existing research on wireless caching and delivery has focused only on one-shot requests, the popularity distribution of contents requested consecutively is definitely different from the one-shot request and has been not considered. 
		Also, especially from the perspective of the service provider, it is advantageous for users to consume as many contents as possible. 
		Thus, this paper proposes two cache allocation policies for categorized contents and consecutive user demands, which maximize 1) the cache hit rate and 2) the number of consecutive content consumption, respectively. 
		Numerical results show how categorized contents and consecutive content requests have impacts on the cache allocation.
	\end{abstract}
	
	\IEEEpeerreviewmaketitle
	
	\section{Introduction}\label{sec:intro}
	
	In multimedia services, e.g., on-demand streaming services, a relatively small number of popular contents generally occupies a large portion of the massive global data traffic, and most of user demands are overlapped and repeated \cite{youtube}.
	To deal with this issue, wireless caching technologies have been studied, wherein the base station (BS) or the server pushes popular contents for off-peak hours to cache-enabled nodes so that these nodes provide contents directly to nearby mobile users during peak hours~\cite{femtocaching}.
	In practice, caching nodes (i.e., caching helpers and/or cache-enabled devices) have finite storage sizes, which leads the content placement problem to determine which content is better to be stored in caching nodes \cite{TIT2013shanmugam}.
	
	The goal of the content placement problem is to find optimal caching policies according to the popularity distribution of contents and network topology.
	In stochastic wireless caching networks, there exist research efforts on probabilistic content placement introduced in \cite{ICC2015blaszczyszyn,CL2017chen}. 
	Many probabilistic caching methods have been proposed for various systems, e.g., device-to-device (D2D) networks \cite{JSAC2016ji}, $N$-tier hierarchical networks \cite{TWC2018Li}, multi-quality dynamic streaming \cite{JSAC2018Choi}, and probabilistic coded caching was also recently proposed in \cite{TWC2019Ko}.
	
	Previous research optimized content placement for users requesting only one content, i.e., one-shot request, under the assumption that all content requests are independent. 
	In multimedia services, the user typically accesses a service platform with the purpose of consuming a specific category of the content, and generally
	consumes more than one content. 
	In this case, the sequence of consecutively consumed contents is highly correlated. 
	For example, in video services such as YouTube, the related video list is recommended to the user after the first content is consumed \cite{ToN2009Cha}.
	In particular, the view count of a given video varies in almost the same scale as the average view count of the top referrer videos in the related list \cite{IMC2010Zhou}; therefore, a user is highly likely to request one of the videos in the related category.
	Therefore, this paper considers the scenario in which a user consecutively requests multiple contents that are likely to be in its preferred category.
	In this context, this paper proposes two cache allocation policies for categorized contents and consecutive user demands, which maximize the cache hit rate and the expected number of consecutive content consumptions, respectively.
	
	The previous work of \cite{ICC2019Choi} has proposed a caching policy for consecutive user demands with the assumption that the number of content requests is fixed. 
	This assumption does not allow the service provider to maximize user's content consumption;
	however, in the perspective of service providers, it is important to satisfy as many of the user's requests as possible.
	In this paper, each user determines whether to continue to consume more contents depending on cache states in its vicinity, and the service provider aims at making users stay in the service longer.
	
	The main contributions are as follows:
	\begin{itemize}
		\item Different from most existing results on wireless caching in which one-shot requests are considered only, consecutive requests of categorized contents are considered.
		In practice, the sequence of consecutively consumed contents is highly correlated, and an advanced caching scheme is required.
		
		\item Based on real data set, the recent work of \cite{ToN2019Lee} has modeled the category-based conditional content popularity distribution. 
		This paper uses this measurement-based popularity model to obtain the proposed cache allocation rule for consecutive requests of categorized contents.
		
		\item This paper proposes two cache allocation schemes which maximize cache hit rates and the expected number of consecutive content requests, respectively.
		The iterative algorithm is presented to find the optimal cache allocation rules and its convergence is proved.
		
		\item Numerical results show how 1) the popularity concentration to the preferred category and 2) different numbers of contents in the different categories influence the cache allocation rule.
	\end{itemize}
	
	The rest of the paper is organized as follows.
	The system model is described in Section \ref{sec:system_model}. 
	The cache allocation rules maximizing the cache hit probability and the number of consecutive content consumption are proposed in Sections \ref{sec:max_cache_hit_prob} and \ref{sec:max_video_number}, respectively.
	The numerical results are shown in Section \ref{sec:numerical_results} and Section \ref{sec:conclusion} concludes the paper.
	
	\section{System Model}
	\label{sec:system_model}
	
	\subsection{Wireless Caching Network}
	
	This paper assumes that caching nodes are randomly distributed according to a general spatial distribution $\Phi$, and the server which has a content library $\mathcal{N}$ pushes some popular contents to each caching node during off-peak hours. 
	Suppose that a library $\mathcal{N}$ consists of $N$ contents and all contents have a normalized unit size. 
	Let all $N$ contents be grouped into $K$ categories, and $N_i$ contents are in category $i$ denoted by $\mathcal{C}_i$, for all $i\in\mathcal{K} = \{1,\cdots,K \}$, satisfying $\sum_{i=1}^K N_i = N$.
	Also, denote the content index set of $\mathcal{C}_i$ by $\mathcal{N}_i = \{1,\cdots,N_i \}$.
	
	The caching nodes have the finite storage size $M$, which means only $M$ contents can be cached in each node. 
	Since $N>M$ in practice, caching nodes store a part of contents in $\mathcal{N}$.
	A user requests the content from caching nodes in its vicinity. 
	If the user finds at least one caching node that stores the desired content, this case is called the cache hit. 
	When the user requests multiple contents, we define the cache hit as the case where {\em all} of requested contents can be found in nearby caching nodes. 
	When there is no caching node having some of the requested contents, the server can deliver them via a cellular link.
	However, this paper assumes that the transmission quality of the cellular link is insufficient due to delay and/or congestion that leads to unacceptable video quality, so that henceforth we do not consider direct transmission from the server.
	
	Let the storage size $M$ be divided into $K$ fractions with unequal sizes denoted by $\alpha_i$ for all $i\in\mathcal{K}$, and contents in $\mathcal{C}_i$ are stored within $\alpha_i$. 
	These fractions will be called cache allocations for categories and satisfy $\sum_{i=1}^K \alpha_i \leq M$ and $\alpha_i \leq N_i$.
	Given all of $\boldsymbol{\alpha}=(\alpha_1,\cdots,\alpha_K)$, how to store individual contents within each category becomes a classical content placement problem, and we consider the probabilistic caching policy for individual contents as shown in \cite{ICC2015blaszczyszyn,CL2017chen}.
	
	\begin{table}[t!]%
		\small
		\caption{Key notations}
		\label{table:parameters}
		\vspace{-3.5mm}
		{\footnotesize
			\begin{center}
				\begin{tabular}{|c|l|}
					\hline
					$K$ & Number of categories \\
					\hline
					$N_i$ & Number of contents in category $i$ \\
					\hline
					$k$ & Index of the preferred category \\
					\hline
					$M$ & Cache size \\ 
					\hline
					$\alpha_i$ & Cache allocation for category $i$ \\ 
					\hline 
					$f_i$ & Global popularity of category $i$ \\
					\hline
					$p_1$ & Popularity of the preferred category \\
					\hline
					$r$ & Rank of category that requests contents \\
					\hline
					$l$ & Number of consecutive content requests \\ 
					\hline
					$l_r$ & Number of requested contents in the $r$-ranked category \\
					\hline
					$i_r$ & Category index of the $r$-ranked category \\
					\hline
					$\epsilon$ & Probability of not requesting the next content \\
					\hline
				\end{tabular}
			\end{center}
		}
	\vspace{-5mm}
	\end{table}

	\subsection{Content Popularity Model}
	
	This paper focuses on the scenario in which the user requests multiple contents consecutively, different from most of existing caching policies which considered only one-shot requests.
	A representative example is a video streaming service. 
	For example, a user can access the service platform with a concrete purpose of watching some sports highlight clips.
	In this case, we can postulate that sports is the user's preferred category, therefore the probability of requesting sports videos in sequence is very high.
	In contrast, the probability of requesting contents in other categories, e.g., movie trailers, is very small although not zero. 
	
	Accordingly, the content request can be modeled by the following steps. 
	First, the user randomly picks one category in $\mathcal{K}$. 
	Each category $i \in \mathcal{K}$ has a global category popularity, which follows the Zipf distribution \cite{ICC2015blaszczyszyn}: $f_i = i^{-\gamma}/\sum_{j=1}^K j^{-\gamma}$  where $\gamma$ denotes the popularity distribution skewness.
	Then, the selected category has the first rank among all categories; note that the global category popularity is only used for choosing the first rank. 
	Other categories can have any rank except for the first rank, but this paper models all categories that are not the first for this particular user as statistically equivalent; in other words, the relative ranking from $2,\cdots,K$ does not matter.
	After determining the preferred category, the user chooses one of categories to request the content depending on their ranks. 
	Again, the category rank distribution given the preferred category is assumed to follow the Zipf distribution, i.e., $\mathrm{Pr}\{R=r \} = r^{-\gamma^{\text{out}}}/\sum_{j=1}^K j^{-\gamma^{\text{out}}}$, which represents the popularity of the $r$-th ranked category and $\gamma^{\text{out}}$ is the Zipf factor.
	We denote the popularity of the preferred category by $p_1 = \mathrm{Pr}\{R=1\}$. 
	Note that $p_1$ is the probability of staying within the given preferred category, not the general probability of picking the 1st-ranked category as in \cite{ToN2019Lee}, which is a different quantity. 
	Here, we also consider the situation in which the user can stop to request contents by itself with small probability of $\epsilon$.
	Therefore, the probability of requesting any content in the $r$-th ranked category after consuming the first content becomes $(1-\epsilon) \cdot \mathrm{Pr}\{R=r \}$. 
	
	After choosing the category rank, the user requests the specific content in the category having the chosen rank.
	According to \cite{ToN2019Lee}, 
	the category-based conditional popularity distribution of contents in $\mathcal{C}_{i}$ follows the Mandelbrot-Zipf (M-Zipf) distribution, 
	i.e., 
	\begin{equation}
	a_{i,n} = \frac{\frac{1}{(n+c^{\text{in}})^{\gamma_i^{\text{in}}}}}{ \sum_{m=1}^{N_i} \frac{1}{(m+c^{\text{in}})^{\gamma_i^{\text{in}}}} },
	\end{equation} 
	which represents the popularity of the $n$-th content in $\mathcal{C}_i$ for $n\in \mathcal{N}_i$.
	$\gamma_i^{\text{in}}$ and $c^{\text{in}}$ are the Zipf factor and the plateau factor of $\mathcal{C}_i$, respectively. 
	Here, if $\gamma^{\text{out}}$ is sufficiently large, $p_1 \gg 1-\epsilon-p_1$ and the popularity of contents in the preferred $\mathcal{C}_{k}$ is much larger than that of any content in $\mathcal{C}_i$ for all $i\neq k$. 
	Fig. \ref{fig:category-based popularity} shows popularity distribution of 100 contents given the preferred category, grouped into 5 categories consisting of 20 contents. 
	This figure is obtained by multiplying the rank probability and the category-based individual content popularity, when $\gamma^{\text{out}}=5$, $\gamma_i^{\text{in}} = 2.4$ and $c^{\text{in}}=69$.
	Among them, contents whose indices are from 1 to 20 belong to the first-ranked category, and their popularity is relatively much larger than others.
	Therefore, if $\gamma^{\text{out}}$ is sufficiently large, we approximate the popularity of contents outside the preferred $\mathcal{C}_{k}$ as uniform distribution, i.e., $a^{k}_{i,n} \approx \frac{1}{N-N_{k}}$ for all $i \in\mathcal{K} \setminus \{k\}$ and $n\in \mathcal{N}_i$, irrespective of ranks of those categories.
	When given $\mathcal{C}_k$, the popularity of contents in $\mathcal{C}_k$ is $a^k_{k,n} = a_{k,n}$ for $n\in\mathcal{N}_k$ still. 
	Thus, consideration of two exclusive sets of the preferred category and all other contents is reasonable.
	
	\section{Maximization of Cache Hit Probability}
	\label{sec:max_cache_hit_prob}
	
	This section derives the cache hit probability and proposes a cache allocation rule that maximizes the cache hit probability.
	
	\subsection{Cache hit probability}
	
	Suppose that the user request $l$ contents in sequence. 
	Among $l$ contents, let $l_{r}$ contents belong to the $r$-th ranked category satisfying $\sum_{r=1}^K l_{r} = l$. 
	Then, when the preferred category $\mathcal{C}_{k}$ is given, the cache hit probability given $l$ content requests, i.e., the probability that all of $l$ requested contents can be delivered from any caching node, can be expressed as 
	\begin{equation}
	p_{\text{hit}}^{l,k} = \prod_{r=1}^{K} \Big[\mathrm{Pr}\{R=r \} h^k_{i_r}(\alpha_{i_r}) \Big]^{l_{r}},
	\end{equation}
	where $i_r \in \mathcal{K}$ is the index of the $r$-th ranked category and 
	\begin{equation}
	h^{k}_i(\alpha_i) = 1 - \sum_{n=1}^{N_i} a^{k}_{i,n} \sum_{j=0}^{\infty} \mathrm{Pr}_{\Phi}\{J=j \} (1-b_{i,n}(\alpha_i) )^j \label{eq:h_k},
	\end{equation}
	is the cache hit probability of a content request within $\mathcal{C}_i$ given $\alpha_i$ when $\mathcal{C}_{k}$ is the preferred category. 
	In Eq. \eqref{eq:h_k}, $\mathrm{Pr}_{\Phi}\{J=j \}$ is the probability that there are $j$ caching nodes storing the requested content in the vicinity of the user.
	Also, $b_{i,n}(\alpha_i)$ is the caching probability of the $n$-th content in $\mathcal{C}_i$ given $\alpha_i$. 
	However, computations required for scanning all combinations of $l
	_{r}$ values are exponentially increasing as $l$ grows.
	
	When $\gamma^{\text{out}}$ is large, $p_{\text{hit}}^{l,k}$ is simplified by using approximations of $a^{k}_{i,n} \approx \frac{1}{N-N_{k}}$, $\forall i \in\mathcal{K} \setminus \{k\}$ and $n\in \mathcal{N}_i$ into
	\begin{equation}
	p_{\text{hit}}^{l,k} \approx \sum_{m=0}^{l} \binom{l}{m} [p_1 h^{k}_{k}(\alpha_{k})]^m \cdot [(1-\epsilon-p_1) q_{k}(\boldsymbol{\alpha})]^{l-m}, \label{eq:cache hit given lk}
	\end{equation}
	where 
	\begin{align}
	q_{k}(\boldsymbol{\alpha}) = 1 - \sum_{\substack{i=1 \\ i \neq k}}^{K} \sum_{n=1}^{N_i} \frac{1}{N-N_{k}} \sum_{j=0}^{\infty} \mathrm{Pr}_{\Phi}\{J=j \} (1 - b_{i,n}(\alpha_i))^j \label{eq:q_k}
	\end{align}
	which is the cache hit probability of a content request outside $\mathcal{C}_{k}$.
	Each term in Eq. \eqref{eq:cache hit given lk} is the probability that among $l$ requested contents, $m$ are in $\mathcal{C}_{k}$ and $l-m$ contents are outside $\mathcal{C}_{k}$, and all of $l$ contents can be found in caching nodes in the vicinity of the user. 
	For simplicity, we will use the notation $h_{k}(\alpha_{k}) = h_{k}^{k}(\alpha_{k})$ in the following sections.
	
	The expected cache hit probability can be finally derived as 
	\begin{align}
	p_{\text{hit}} &= \sum_{k=1}^{K} f_{k} \sum_{l=1}^{\infty} \mathrm{Pr}\{L=l \} \cdot \sum_{m=0}^{l} \binom{l}{m} [p_1 h_{k}(\alpha_{k})]^m \nonumber \\ 
	&~~~~~\cdot [(1-\epsilon-p_1) q_{k}(\boldsymbol{\alpha})]^{l-m},
	\end{align}
	where $\mathrm{Pr}\{L=l \}$ is the probability that the user requests $l$ contents in sequence, which is given by 
	\begin{equation}
	\mathrm{Pr}\{L=l \} = \epsilon (1-\epsilon)^l.
	\end{equation}
	Therefore, the cache hit probability is arranged into
	\begin{align}
	p_{\text{hit}} &= \sum_{k=1}^{K} f_{k} \sum_{l=1}^{\infty} \epsilon (1-\epsilon)^l (p_1 h_{k}(\alpha_{k}) + (1-\epsilon - p_1) q_{k}(\boldsymbol{\alpha}) )^l \\
	&= \sum_{k=1}^{K} \frac{f_k (1-\epsilon)\epsilon (p_1 h_k(\alpha_k) + (1-\epsilon-p_1) q_k(\boldsymbol{\alpha})) }{1 - (1-\epsilon) (p_1 h_k(\alpha_k) + (1-\epsilon-p_1) q_k(\boldsymbol{\alpha}))}
	\label{eq:cache_hit}
	\end{align}
	
	\begin{figure} [t!]
		\centering
		\includegraphics[width=0.28\textwidth]{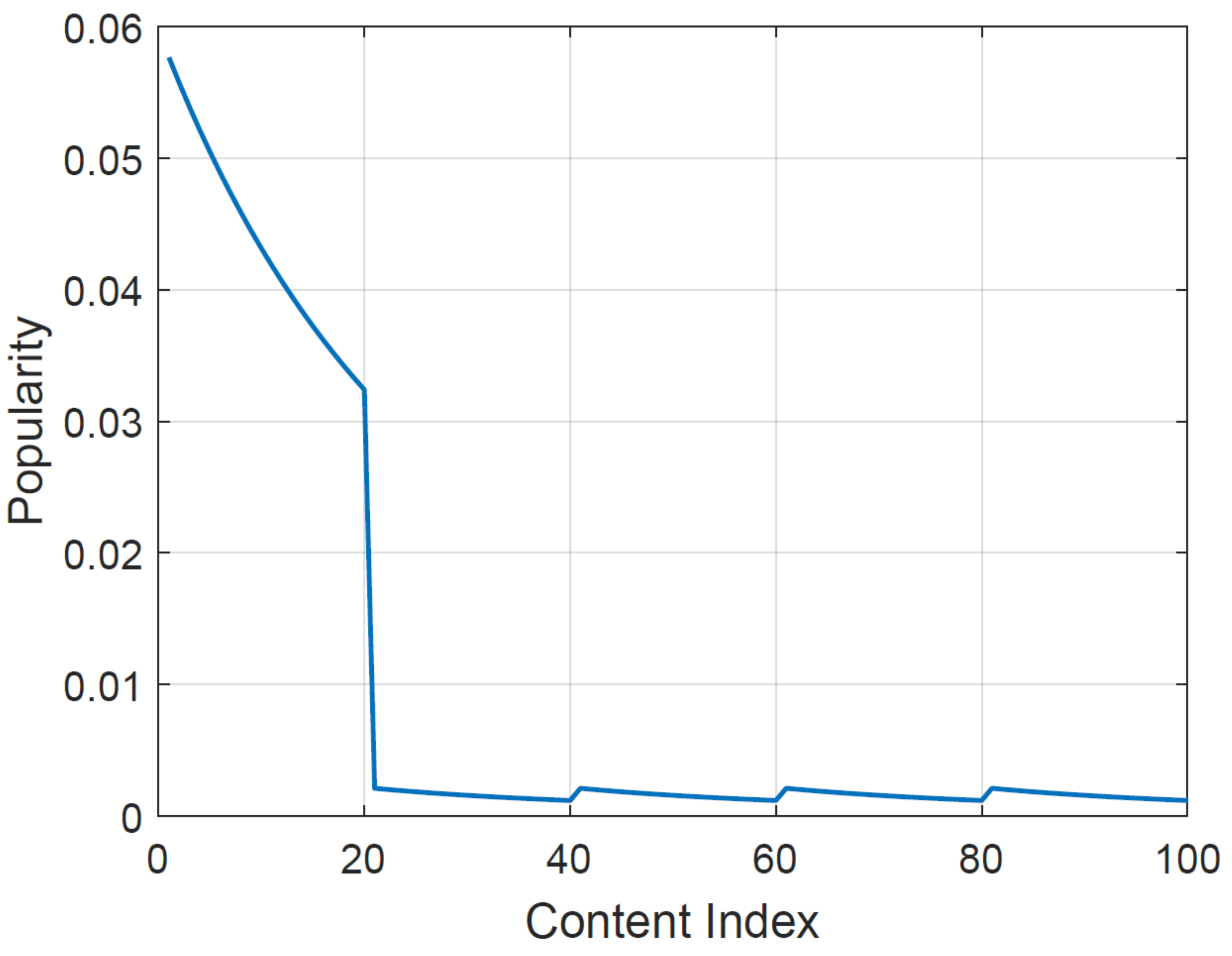}
		\caption{Popularity of contents given the preferred category}
		\label{fig:category-based popularity}
		\vspace{-2mm}
	\end{figure}
	
	In Eq. \eqref{eq:cache_hit}, any caching policy can be utilized within each category given the preferred $\mathcal{C}_k$ and $\boldsymbol{\alpha}$, and $h_k(\alpha_k)$ and $q_k(\boldsymbol{\alpha})$ are determined depending on the caching policy.
	Then, we can suppose that the caching policy that maximizes the cache hit rate \cite{ICC2015blaszczyszyn,CL2017chen} is used for caching of individual contents within every category.
	Denote the maximum cache hit rates in and outside $\mathcal{C}_k$ by $h_k^*(\alpha_k)$ and $q_k^*(\boldsymbol{\alpha})$, respectively.
	Then, the cache hit probability of $l$ content requests becomes
	\begin{align}
	p^*_{\text{hit}} &= \sum_{k=1}^{K} \frac{f_k (1-\epsilon)\epsilon (p_1 h^*_k(\alpha_k) + (1-\epsilon-p_1) q^*_k(\boldsymbol{\alpha})) }{1 - (1-\epsilon) (p_1 h^*_k(\alpha_k) + (1-\epsilon-p_1) q^*_k(\boldsymbol{\alpha}))} \\
	&= \sum_{k=1}^{K} g^*_k(\boldsymbol{\alpha}, f_k, \epsilon). \label{eq:cache_hit_final}
	\end{align}
	
	\subsection{Problem Formulation}
	
	The optimal cache allocations of $\boldsymbol{\alpha}^{\star}=(\alpha_1^{\star}, \cdots, \alpha_K^{\star} )$ can be obtained by maximizing \eqref{eq:cache_hit_final} as follows:
	
	\begin{align}
	&\boldsymbol{\alpha}^{\star} = \underset{\boldsymbol{\alpha}}{\arg\max}~~p^*_{\text{hit}} \label{eq:opt1_obj} \\
	&~~~\text{s.t.}~\sum_{i=1}^K \alpha_i \leq M \label{eq:opt1_const1} \\
	&~~~~~~~~0\leq \alpha_i \leq \min\{M, N_i\},~\forall i\in\mathcal{K}. \label{eq:opt1_const2}
	\end{align}
	
	The constraint \eqref{eq:opt1_const1} is for the storage size of each caching node and the constraint \eqref{eq:opt1_const2} is for the cache allocation for each category.
	The following key lemmas are used to solve the above problem of \eqref{eq:opt1_obj}--\eqref{eq:opt1_const2}.
	
	\begin{lemma}
		$p_1 h_k(\alpha_k) + (1-\epsilon-p_1) q_k(\boldsymbol{\alpha})$ is increasing with $\alpha_k$.
		\label{lemma1}
	\end{lemma}
	\begin{proof}
		In this proof, we simply use the notation of $b_{i,n} = b_{i,n}(\alpha_i),~\forall i\in \mathcal{K},~\forall n\in \mathcal{N}_i$. 
		Let $\alpha'_k = \alpha_k + \delta$ and $\delta > 0$. 
		Then, $h^*_k(\alpha'_k) \geq \tilde{h}_k(\alpha'_k)$, where $\tilde{h}_k (\alpha'_k)$ can be the cache hit probability within $\mathcal{C}_k$ of any  caching policy $\mathbf{b}'_k = (b'_{k,1}, \cdots, b'_{k,N_k})^T$ satisfying $\sum_{n=1}^{N_k} b'_{k,n} = \alpha'_k$.
		Let $\mathbf{b}'_k = (b'_{k,1}=b^*_{k,1}+\delta, b'_{k,2} = b^*_{k,2} \cdots, b'_{k,N_k} = b^*_{k,N_k})^T$.
		Then, since $0 \leq b_{k,1}^* \leq 1$ and $b_{k,1}^*$ is generally much closer to zero than one when the library size of $N$ is large,
		\begin{align}
		&p_1 h^*_k(\alpha'_k) - p_1 h_k^*(\alpha^*_k) \geq p_1 \tilde{h}_k(\alpha'_k) - p_1 h_k^*(\alpha^*_k) \nonumber \\ 
		&=p_1 \sum_{n=1}^{N_k} a_{k,n} \sum_{j=0}^{\infty} \mathrm{Pr}_{\Phi}\{J=j \} \{ (1-b^*_{k,n})^j - (1-b'_{k,n})^j \} \\
		&\approx p_1 \cdot a_{k,1}  \sum_{j=0}^{\infty} \mathrm{Pr}_{\Phi}\{J=j \} \cdot j\cdot \delta \label{eq:p_kh_k_last}
		\end{align}
		is obtained by using the first-order Taylor approximation, i.e., $(1-b^*_{k,n})^j \approx 1-j\cdot b^*_{k,n}$.
		
		Since the storage size $M$ is fixed, the cache size allocated to all categories except for $\mathcal{C}_k$ is $M-\alpha_k-\delta$.
		With small $\delta$, there exists any category $u$ for $u\neq k$ such that $\alpha_u \geq \delta$.
		Then, let $\alpha'_u = \alpha^*_u - \delta$ and $b'_{u,1} = b^*_{u,1} - \eta_1$, $b'_{u,2} = b^*_{u,2} - \eta_2$, $\cdots$, $b'_{u,N_u} = b^*_{u,N_u} - \eta_{N_u}$, where $0 \leq \eta_{u,n} \leq b^*_{u,n}$ for all $n\in \mathcal{N}_u$ and $\sum_{n=1}^{N_u} \eta_{u,n} = \delta$.
		In this case, cache allocations for other categories can remain unchanged, i.e., $\alpha'_i = \alpha^*_i$ and $b'_{i,n} = b^*_{i,n}$ for all $i \in \mathcal{K},~i\neq k,u$.
		Then, similar to before,
		\begin{align}
		&(1-\epsilon-p_1) q^*_k(\boldsymbol{\alpha'}) - (1-\epsilon-p_1) q^*_k(\boldsymbol{\alpha}) \geq  \nonumber \\
		&~ \frac{1-\epsilon-p_1}{N-N_k} \bigg[ \sum_{\substack{i=1\\i\neq k}}^{K} \sum_{n=1}^{N_i} \sum_{j=0}^{\infty} \mathrm{Pr}\{ J=j \} \times \nonumber \\ 
		&~~~~~~~~~~~~~~~~~~~~~~\Big\{ ( 1- b_{i,n}^* )^j - ( 1- b'_{i,n} )^j \Big\}   \bigg] \\
		&~\approx -\frac{1-\epsilon-p_1}{N-N_k} \sum_{j=0}^J \mathrm{Pr}\{ J=j \} \cdot j \cdot \delta.
		\end{align}
		
		Since $p_1 > 1-\epsilon-p_1$ and $a_{k,1} > \frac{1}{N-N_k}$, $p_1 h^*_k(\alpha'_k) + (1-\epsilon-p_1) q^*_k(\boldsymbol{\alpha}') - p_1 h^*_k(\alpha_k^*) - (1-\epsilon-p_1) q^*_k(\boldsymbol{\alpha}^*) > 0$ and the above lemma is finally proved.
	\end{proof}
	
	\begin{algorithm}[t!]
		\caption{Greedy cache allocation algorithm
			\label{algo}}
		\begin{algorithmic}[1]
			\State{$\alpha^*_i = \frac{M}{K}$ for all $i \in\mathcal{K}$ and $p_{\text{hit}}^{\star} = 0$}
			\For{$\forall (u,v) \in \mathcal{K}\times \mathcal{K}$ and $u\neq v$}{
				\State{$\beta_{u,v} = M-\alpha^*_u - \alpha^*_v$}
				\State{Obtain $b_{i,n}~\forall i\in\mathcal{K}$ and $i\neq u,v$, and $\forall n \in \mathcal{N}_i$ according to \cite{ICC2015blaszczyszyn,CL2017chen}.}
				\For{$\forall \alpha_u \in \{ \max(0, \beta_{u,v}-N_v), \cdots, \min(\beta_{u,v},N_u) \} $}{
					\State{$\alpha_v \leftarrow \beta_{u,v} - \alpha_u$}
					\State{Obtain $b_{u,n}$ and $b_{v,m}$ $\forall n\in \mathcal{N}_u$ and $\forall m\in \mathcal{N}_v$ according to \cite{ICC2015blaszczyszyn,CL2017chen}.}
					\State{Find $p_{\text{hit}}^*$ based on $\boldsymbol{\alpha}$.}
					\If{$p_{\text{hit}}^{\star} < p_{\text{hit}}^*$} $p_{\text{hit}}^{\star} = p_{\text{hit}}^*$ \EndIf
				}
				\EndFor
			}
			\EndFor
		\end{algorithmic}
	\end{algorithm}
	\begin{lemma}
		The optimum vector $\boldsymbol{\alpha}^{\star} = (\alpha_1^{\star}, \cdots, \alpha_K^{\star})^T$ satisfies $\sum_{i=1}^K \alpha_i^{\star} = M$.
		\label{lemma2}
	\end{lemma}
	\begin{proof}
		Assume that $\sum_{i=1}^K \alpha^{\star}_i < M$, then $\exists \delta > 0$ such that $\sum_{i=1}^K \alpha^{\star}_i + \delta \leq M$ and $\alpha^{\star}_k + \delta < \min\{M,N_k\}$ for certain $k$. 
		Let $\boldsymbol{\alpha}' = (\alpha^{\star}_1, \cdots, \alpha_k^{\star}+\delta, \cdots, \alpha_K^{\star} )^T$.
		According to Lemma 1, $g_k^*(\boldsymbol{\alpha},f_k,\epsilon)$ is increasing with $\alpha_k$ for all $k\in\mathcal{K}$.
		Thus, $p_{\text{hit}}^*(\boldsymbol{\alpha}') > p_{\text{hit}}^*(\boldsymbol{\alpha}^{\star})$ and it obviously leads to contradiction.
	\end{proof}
	According to Lemma \ref{lemma2}, an inequality constraint \eqref{eq:opt1_const1} can be converted into the equality constraint.
	The problem of \eqref{eq:opt1_obj}--\eqref{eq:opt1_const2} has $K$ optimization parameters, and the subproblem for finding the optimal $\alpha_u^{\star}$ and $\alpha_v^{\star}$ is formulated as follows:
	\begin{align}
	\{\alpha^{\star}_u, \alpha^{\star}_v \} &= \underset{\alpha_u, \alpha_v}{\arg \max}~\mathcal{M}_{(u, v)} \label{eq:subprob_obj} \\
	&\text{s.t.}~\alpha_u + \alpha_v = \beta_{u,v} = M - \sum_{i=1,i\neq u,v}^{K} \alpha_i \label{eq:subprob_const1} \\
	&~~~~~0\leq \alpha_i \leq \min\{ M, N_i \},~\forall i \in \mathcal{K}, \label{eq:subprob_const2}
	\end{align}
	where $\mathcal{M}_{(u, v)} = g^*_u(\boldsymbol{\alpha}, f_u, \epsilon) + g^*_v(\boldsymbol{\alpha}, f_v, \epsilon)$.
	Since $\{\alpha_i \}_{i\neq u,v}$ are fixed, $\alpha_u+\alpha_v$ also becomes a constant $\beta_{u,v}$.  
	
	A multivariable function $p_{\text{hit}}^*$ can be optimized by iteratively optimizing the subset of variables if the convergence is guaranteed. 
	To find $\boldsymbol{\alpha}^{\star} = (\alpha_1^{\star},\cdots,\alpha_K^{\star} )$, the subproblem of \eqref{eq:subprob_obj}--\eqref{eq:subprob_const2} can be iteratively applied for all combinations of $u$ and $v$, for $u,v \in \mathcal{K}$ and $u\neq v$.
	We find the maximum of the dual-variable problem of \eqref{eq:subprob_obj}--\eqref{eq:subprob_const2} in each iteration, and the sequence of the updated values of $\mathcal{M}_{(u,v)}$ is generated. 
	Since this sequence is non-decreasing and the cache hit probability has a trivial upper bound of 1, i.e., $p_{\text{hit}}^{\star} \leq 1$, the convergence of the iterative algorithm is guaranteed.
	
	Since $h_k(\alpha_k)$ and $q_k(\boldsymbol{\alpha})$ are obtained by using the bisection method \cite{ICC2015blaszczyszyn,CL2017chen}, however, the objective function of $\mathcal{M}_{(u,v)}$ is not in closed-form and the problem of \eqref{eq:subprob_obj}--\eqref{eq:subprob_const2} should be numerically handled.
	Therefore, we consider integer values for cache allocations of $\boldsymbol{\alpha}$ and the greedy algorithm can solve the problem with $M$ not very large. 
	If caching of content partitions is not considered, i.e., only caching of the whole content is allowed, the assumption that $\alpha_i$ is the integer number for all $i\in\mathcal{K}$ is reasonable.
	The details of the iterative algorithm to solve the problem of \eqref{eq:opt1_obj}--\eqref{eq:opt1_const2} are described in Algorithm 1.
	
	\section{Maximization of expected number of consecutive content requests}
	\label{sec:max_video_number}
	
	From the service provider's perspective, it is advantageous for the user to consume as many contents as possible.
	As explained in Section \ref{sec:system_model}, the user does not request the next content with the probability of $\epsilon$. 
	In addition, we assume that the user stops to consume the next content when no caching node in the vicinity of the user stores the desired content, even though the user requests the next one.
	
	The probability of stopping to consume more is given by
	\begin{equation}
	p^k_{\text{stop}} = \epsilon + p_1 (1 - h_k(\alpha_k)) + (1-\epsilon-p_1) (1 - q_k(\boldsymbol{\alpha})). \label{eq:stop_watching}
	\end{equation}
	In \eqref{eq:stop_watching}, the first term is the probability of not requesting the next content, the second and third terms are probabilities that no caching node stores the requested content when the content belongs to $\mathcal{C}_k$ and is not in $\mathcal{C}_k$, respectively.
	Then, the expected number of consecutive content consumption is computed as
	\begin{align}
	\mathbb{E}[L] &= \sum_{l=1}^{\infty} l \cdot \mathrm{Pr}\{L=l\} = \sum_{k=1}^K f_k \sum_{l=1}^{\infty} l \cdot (1-p^k_{\text{stop}})^l p^k_{\text{stop}} \\
	&= \sum_{k=1}^K f_k \frac{1-p^k_{\text{stop}}}{p^k_{\text{stop}}}.
	\end{align}
	
	Then, the optimization problem of maximizing the expected number of consecutive video consumption is as follows:
	\begin{align}
	&\boldsymbol{\alpha}^{\star} = \underset{\boldsymbol{\alpha}}{\arg \max}~\mathbb{E}[L] \label{eq:opt2_obj} \\
	&~~~\text{s.t.}~\sum_{i=1}^K \alpha_i \leq M \label{eq:opt2_const1} \\
	&~~~~~~~~0\leq \alpha_i \leq \min\{M, N_i\},~\forall i\in \mathcal{K}. \label{eq:opt2_const2}
	\end{align}
	Similar to Lemmas \ref{lemma1} and \ref{lemma2}, $p^k_{\text{stop}}$ can be proved to be increasing with $\alpha_k$ and the inequality constraint \eqref{eq:opt2_const1} can be converted into the equality constraint, i.e., $\sum_{i=1}^K \alpha_i = M$.
	
	Again, the multivariable function $\mathbb{E}[L]$ can be maximized by iteratively optimizing the following dual-variable subproblem:
	\begin{align}
	&\{\alpha^{\star}_u, \alpha^{\star}_v \} = \underset{\alpha_u, \alpha_v}{\arg \max}\frac{f_u}{p^u_{\text{stop}}} + \frac{f_v}{p^v_{\text{stop}}} \label{eq:subprob2_obj} \\
	&~~\text{s.t.}~\alpha_u + \alpha_v = \beta_{u,v} = M - \sum_{i=1, i\neq u,v}^{K} \alpha_i \label{eq:subprob2_const1} \\
	&~~~~~~~0\leq \alpha_i \leq \min\{ M, N_i \},~\forall i \in \mathcal{K}. \label{eq:subprob2_const2}
	\end{align}
	
	The sequence of the updated objective values in \eqref{eq:subprob2_obj} is nondecreasing, and $\mathbb{E}[L] \leq \frac{1}{\epsilon}-1$ because $p^k_{\text{stop}} \geq \epsilon$.
	Thus, the algorithm which solves the problem of \eqref{eq:opt2_obj}--\eqref{eq:opt2_const2} by iteratively optimizing the dual-variable problem of \eqref{eq:subprob2_obj}--\eqref{eq:subprob2_const2} for all combinations of $u,v\in\mathcal{K}$ and $u\neq v$, is guaranteed to converge.
	The whole algorithm is the same as Algorithm \ref{algo} except that $p_{\text{hit}}^*$ should be changed into $\mathbb{E}[L]$ in lines 2, 9, 10.
	
	\section{Numerical Results}
	\label{sec:numerical_results}
	
	In the subsequent simulations, $N=100$ contents and $K=5$ categories are considered. 
	The global category popularity follows $f_i > f_j$ for $i<j$.
	Caching nodes are distributed according to a Poisson point process with intensity of $\lambda$ and caching nodes with distances less than $d=10$ from the user are only considered.
	In addition, $M=30$, $\gamma=1$, $\gamma_i^{in}=2.4$, and $c_i^{\text{in}}=69$ are used for all $i\in \mathcal{K}$.
	We consider three different category structures as follows:
	\begin{itemize}
		\item Case A: $N_1=N_2=N_3=N_4=N_5=20$ 
		\item Case B: $N_1=35, N_2=25, N_3=20, N_4=15, N_5=5$ 
		\item Case C: $N_1=5, N_2=15, N_3=20, N_4=25, N_5=35$.
	\end{itemize}
	
	\begin{figure} [t!]
		\centering
		\includegraphics[width=0.3\textwidth]{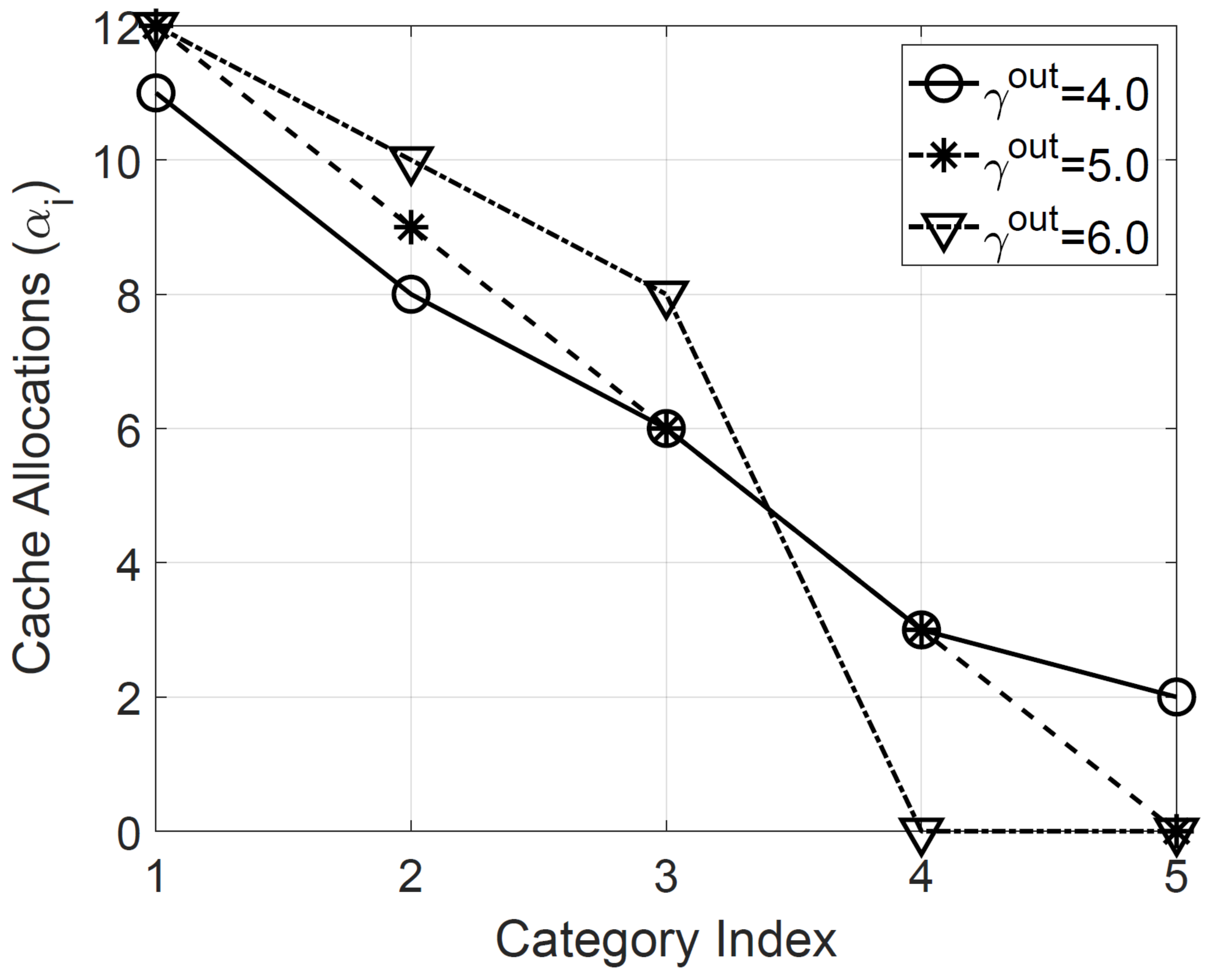}
		\caption{Cache allocations depending on the skewness factor}
		\label{fig:cache_allocation_gam}
	\end{figure}
	
	\begin{figure} [t!]
		\centering
		\includegraphics[width=0.3\textwidth]{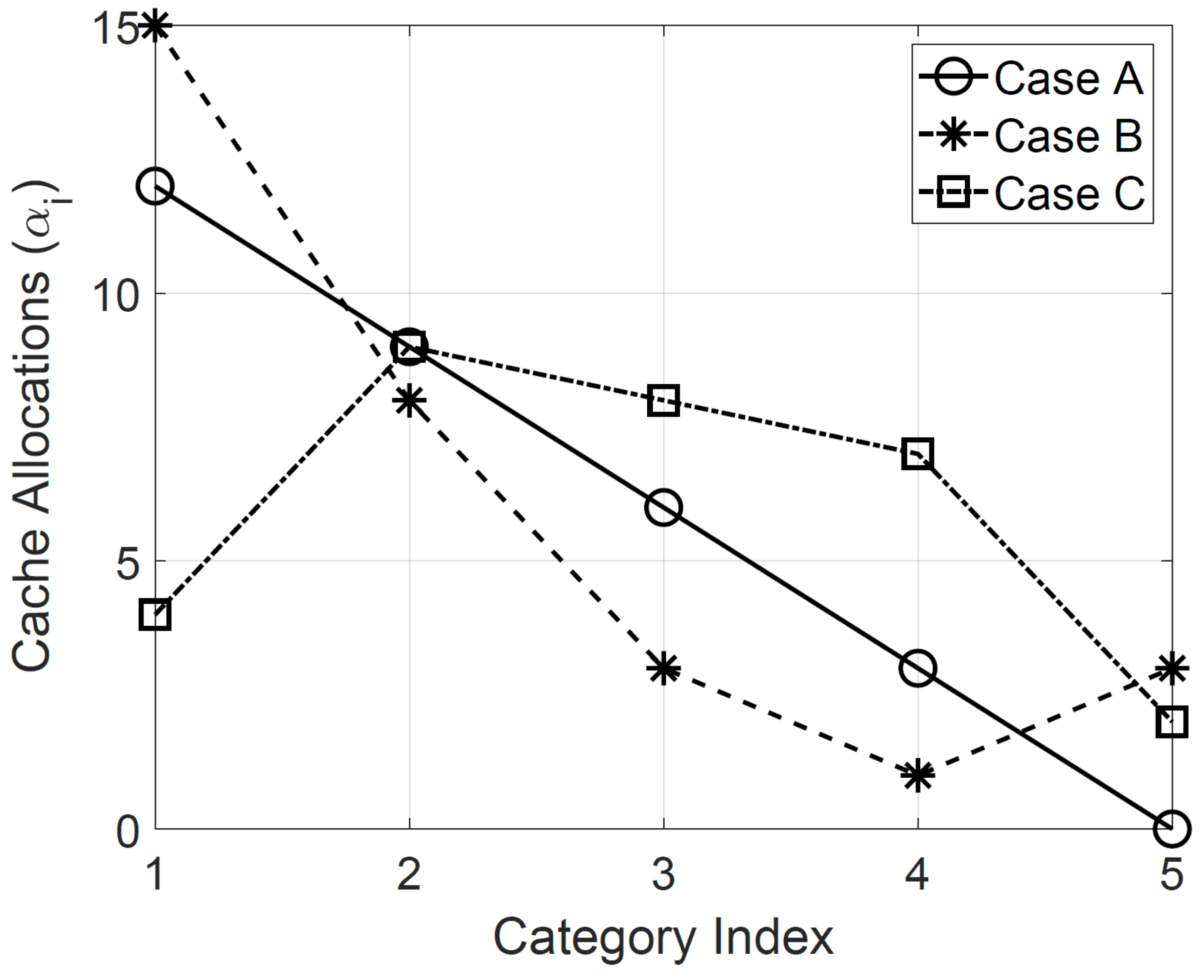}
		\caption{Cache allocations depending on the number of contents in each category}
		\label{fig:cache_allocation_case}
	\end{figure}
	
	In Figs. \ref{fig:cache_allocation_gam} and \ref{fig:cache_allocation_case}, plots of $\alpha_k$ for every category are shown with $\lambda=0.02$ and $\epsilon=0.1$. 
	Case A is considered in Fig. \ref{fig:cache_allocation_gam}.
	As the skew factor $\gamma^{\text{out}}$ grows, the probability of requesting the content in the preferred category becomes much larger than that of requesting the content in other categories.
	Therefore, as $\gamma^{\text{out}}$ increases, more cache sizes are allocated to categories having relatively large global category popularities in Fig \ref{fig:cache_allocation_gam}.
	
	In Fig. \ref{fig:cache_allocation_case}, all plots are obtained with $\gamma^{\text{out}}=5$.
	Since all categories in Case A have the same number of contents, cache allocations of Case A depend only on global category popularity. 
	In Case B, the category having a larger global popularity consists of more contents, therefore more cache sizes are allocated, i.e., $\alpha_1$ in Case B becomes larger than that in Case A. 
	Interestingly, $\alpha_5$ in Case B is also larger than that in Case A. 
	The reason is that $N_5$ is the smallest in Case B, i.e., the individual content popularity within $\mathcal{C}_5$ is the largest among all categories. 
	Thus, even though $f_5$ is smaller than other $f_i$ values, caching multiple contents of $\mathcal{C}_5$ is favorable for consecutive content requests. 
	On the other hand, in Case C, $\alpha_1$ is smaller than $\alpha_2$, $\alpha_3$ and $\alpha_4$ because $N_1=5$ and $\alpha_1$ should be smaller than $N_1$.
	It does not mean that an importance of caching contents in $\mathcal{C}_1$ decreases.
	Rather, it becomes more important because a portion of contents to be stored in caching nodes, i.e., $\frac{\alpha_1}{N_1}$, is larger than other cases.
	By saving the cache size for $\mathcal{C}_1$, a larger cache size can be allocated to other categories with low global popularities compared to Case A.
	Thus, Figs. \ref{fig:cache_allocation_gam} and \ref{fig:cache_allocation_case} show that the skew factor as well as the number of contents in each category have a strong impact on the proposed cache allocation rule.
	
	Fig. \ref{fig:hit_lamList} shows plots of cache hit probabilities obtained from the problem of \eqref{eq:opt1_obj}--\eqref{eq:opt1_const2} versus $\lambda$. 
	In Fig. \ref{fig:L_lamList}, the expected numbers of consecutive content consumption obtained from the problem of \eqref{eq:opt2_obj}--\eqref{eq:opt2_const2} are shown.
	We compared the proposed scheme with the conventional caching method optimized for one-shot content request based on popularity of individual contents in \cite{ICC2015blaszczyszyn,CL2017chen}.
	The comparison scheme is named as `L1' in the figures.
	We can easily see in both figures that the proposed scheme outperforms `L1' with different valuess of $\gamma^{\text{out}}$ and $N_i$ for each category. 
	As $\lambda$ grows, i.e., as the number of caching nodes in the vicinity of the user grows, the performance improvement of the proposed scheme decreases, because the user becomes more likely to find caching nodes to deliver multiple requested contents even with `L1'.
	The performance gain of the proposed scheme over `L1' is guaranteed when $\gamma^{\text{out}}$ is large.
	Especially in Fig. \ref{fig:L_lamList}, when $\epsilon=0.1$, $\epsilon$ dominates the term in \eqref{eq:stop_watching} representing the probability of stopping to consume contents; therefore, the advantage of the proposed scheme is not remarkable.
	As $\epsilon$ becomes smaller, however, the proposed algorithm is more advantageous for consecutive content consumption than `L1'.
	Thus, the service provider can create the opportunity for users to consume more contents and to stay in the service longer by using the proposed scheme.
	
	\begin{figure} [t!]
		\centering
		\includegraphics[width=0.32\textwidth]{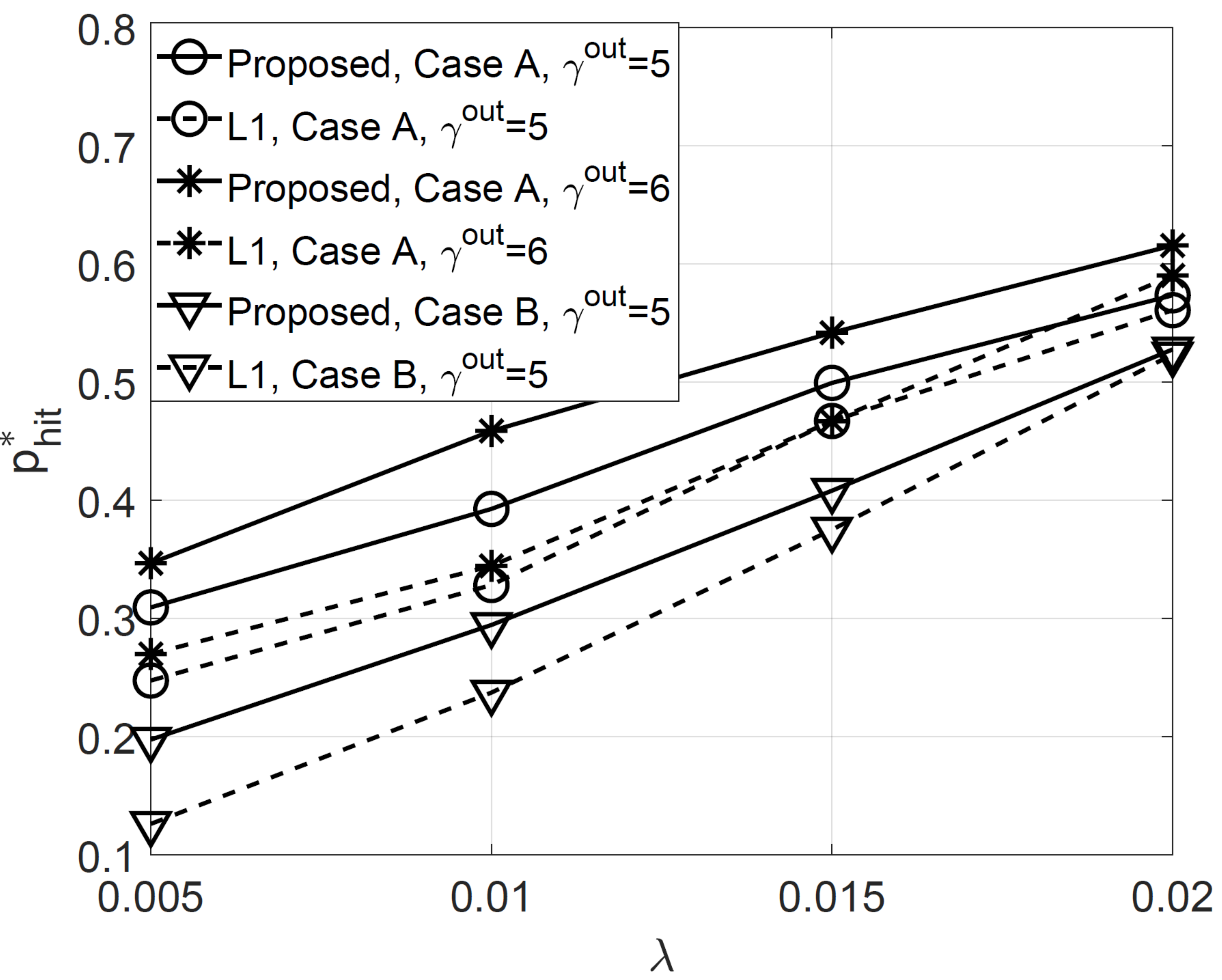}
		\caption{The expected cache hit probabilities}
		\label{fig:hit_lamList}
	\end{figure}
	
	\begin{figure} [t!]
		\centering
		\includegraphics[width=0.32\textwidth]{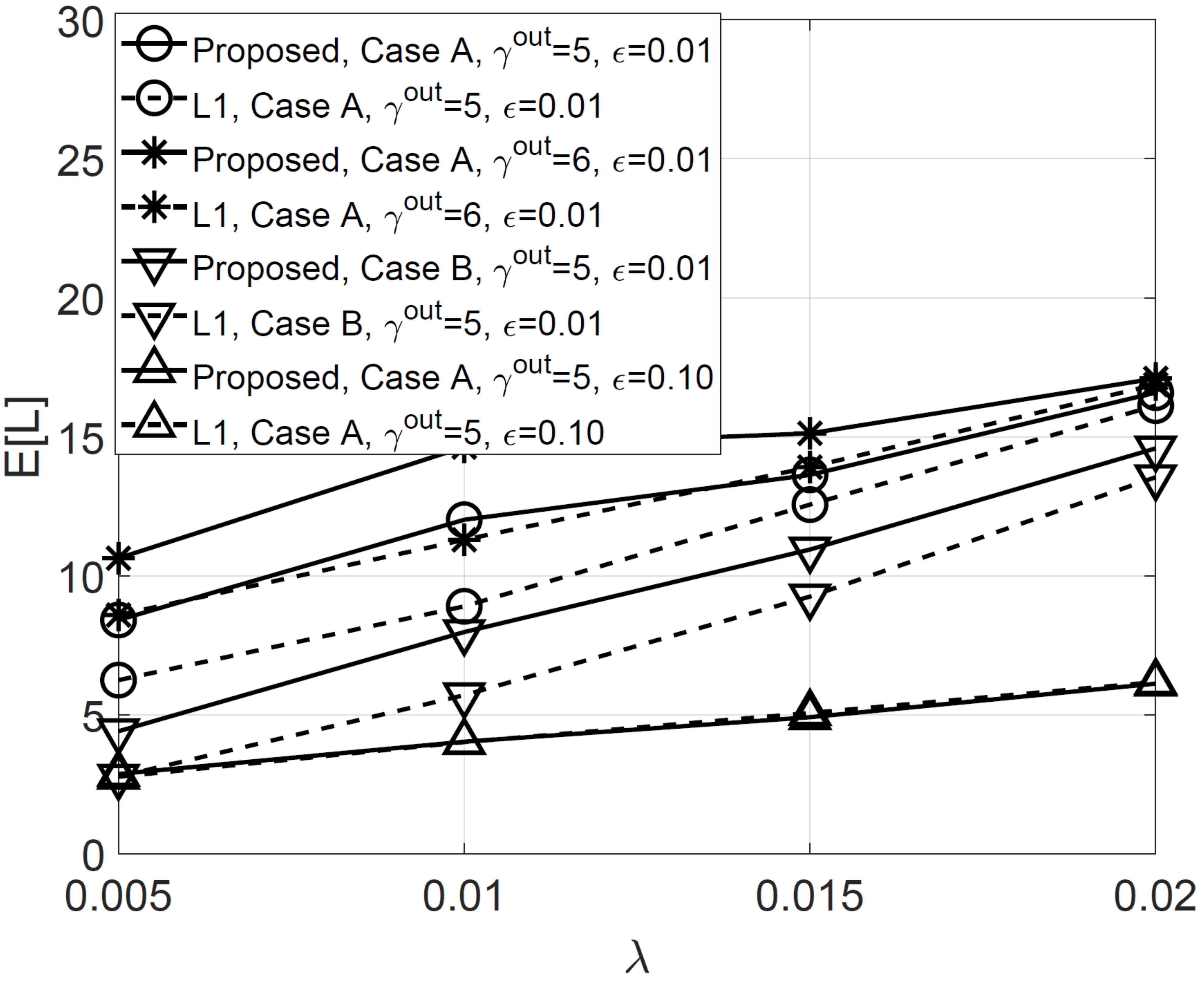}
		\caption{The expected number of consecutive content consumption}
		\label{fig:L_lamList}
	\end{figure}
	
	\section{Concluding Remarks}
	\label{sec:conclusion}
	
	This paper proposes two optimal cache allocation rules when users request a random number of contents consecutively. 
	The key characteristic that users are likely to consume content highly related to each other consecutively is well captured in the proposed scheme by maximizing the cache hit probability for multiple content requests from the same category. 
	Another cache allocation which maximizes the number of consecutive content consumption is also proposed as it related to the benefits for the service providers.
	The impacts of categorized contents and consecutive content requests on the cache allocation rule are shown by numerical results.
	
	\vspace{-1mm}
	\section*{Acknowledgment} 
	This work was supported by NSF under projects NSF
	CCF-1423140 and NSF CNS-1816699, and Institute for Information \&
	Communications Technology Promotion (IITP) grant funded by the Korea government (MSIT) (No.2018-0-00170, Virtual Presence in Moving Objects through 5G).

	
	\vspace{12pt}
	

\begin{thebibliography}{1}
		
		
		\bibitem{youtube}
		X. Cheng, J. Liu, and C. Dale, ``Understanding the Characteristics of Internet Short Video Sharing: A YouTube-based Measurement Study," \textit{IEEE Trans. on Multimedia}, vol. 15, no. 5, pp. 1184--1194, August 2013.
		
		
		\bibitem{femtocaching}
		N. Golrezaei, K. Shanmugam, A. G. Dimakis, A. F. Molisch, and G. Caire, ``FemtoCaching: Wireless Video Content Delivery through Distributed Caching Helpers," in \textit{Proc. IEEE INFOCOM}, Orlando, FL, USA, 2012.
		
		
		
		\bibitem{TIT2013shanmugam}
		K. Shanmugam, N. Golrezaei, A. G. Dimakis, A. F. Molisch, and G. Caire, ``FemtoCaching: Wireless Content Delivery Through Distributed Caching Helpers," \textit{IEEE Trans. on Inf. Theory}, vol. 59, no. 12, pp. 8402--8413, December 2013.
		
		\bibitem{ICC2015blaszczyszyn}
		B. Blaszczyszyn and A. Giovanidis, ``Optimal Geographic Caching in Cellular Networks," in \textit{Proc. IEEE Int'l Conf. on Communications (ICC)}, London, UK, 2015.
		
		\bibitem{CL2017chen}
		Z. Chen, N. Pappas, and M. Kountouris, ``Probabilistic Caching in Wireless D2D Networks: Cache Hit Optimal Versus Throughput Optimal," \textit{IEEE Commun. Letters}, vol. 21, no. 3, pp. 584--587, March 2017.
		
		
		
		\bibitem{JSAC2016ji}
		M. Ji, G. Caire, and A. F. Molisch, ``Wireless Device-to-Device Caching Networks: Basic Principles and System Performance," \textit{IEEE J. Sel. Areas in Commun.}, vol. 34, no. 1, pp. 176--189, Jan. 2016.
		
		\bibitem{TWC2018Li}
		K. Li, C. Yang, Z. Chen and M. Tao, ``Optimization and Analysis of Probabilistic Caching in $N$ -Tier Heterogeneous Networks," \textit{IEEE Trans. Wireless Commun.}, vol. 17, no. 2, pp. 1283-1297, Feb. 2018.
		
		\bibitem{JSAC2018Choi}
		M. Choi, J. Kim, and J. Moon, "Wireless Video Caching and Dynamic Streaming Under Differentiated Quality Requirements," \textit{IEEE J. Sel. Areas in Commun.}, vol. 36, no. 6, pp. 1245--1257, June 2018.
		
		\bibitem{TWC2019Ko}
		D. Ko, B. Hong and W. Choi, ``Probabilistic Caching Based on Maximum Distance Separable Code in a User-Centric Clustered Cache-Aided Wireless Network," \textit{IEEE Trans. Wireless Commun.}, vol. 18, no. 3, pp. 1792-1804, March 2019.
		
		\bibitem{ToN2009Cha}
		M. Cha, H. Kwak, P. Rodriguez, Y. Ahn and S. Moon, ``Analyzing the
		Video Popularity Characteristics of Large-Scale User Generated Content
		Systems," \textit{IEEE/ACM Trans. Network.}, vol. 17, no. 5, pp.
		1357-1370, Oct. 2009.
		
		\bibitem{IMC2010Zhou}
		R. Zhou, S. Khemmarat, and L. Gao, ``The impact of YouTube
		recommendation system on video views" in \textit{Proc. ACM
		IMC}, New York, NY, USA, 404-410, 2010.
		
		\bibitem{ICC2019Choi}
		M. Choi, D. Kim, D.-J. Han, J. Kim and J. Moon, ``Probabilistic Caching Policy for Categorized Contents and Consecutive User Demands", \textit{IEEE Int'l Conf. on Communications (ICC)}, May 2019.
		
		\bibitem{ToN2019Lee}
		M. Lee, A. F. Molisch, N. Sastry and A. Raman, ``Individual Preference Probability Modeling and Parameterization for Video Content in Wireless Caching Networks," \textit{IEEE/ACM Trans. Network.}, 2019.
		
		
	\end{thebibliography}
\end{document}